\documentclass[11pt]{article}

\usepackage{pgf}

\usepackage[english]{babel}
\usepackage[utf8]{inputenc}
\usepackage{lmodern}
\usepackage[T1]{fontenc}
\usepackage{datetime}
\usepackage{longtable}
\usepackage{amsmath,amssymb}
\usepackage{mathtools}
\usepackage{geometry}
\usepackage{array}
\usepackage{xspace}

\usepackage{amsmath,amsfonts,amssymb,amsopn,amsthm}
\theoremstyle{plain}
\usepackage{cases}
\usepackage[all]{xy}
\usepackage[justification=centering]{caption}

\usepackage{url,amssymb}

\newtheorem{proposition}{Proposition}
\newtheorem{theorem}[proposition]{Theorem}
\newtheorem{lemma}[proposition]{Lemma}
\newtheorem{corollary}[proposition]{Corollary}

\theoremstyle{definition}
\newtheorem{definition}[proposition]{Definition}

\theoremstyle{example}
\newtheorem{example}[proposition]{Example}
\newtheorem{remark}[proposition]{Remark}

\newcommand{\gf}{{\mathbb F}}

\makeatother
\newcommand{\GF}[2][2]{{\mathbb F}_{#1^{#2}}}

\newcommand{\Tr}{{\rm Tr}}

\newcommand{\F}{{\mathbb F}}
\newcommand{\C}{{\mathbb C}}
\newcommand{\im}{\textup{Im}}

\def\qu#1 {\fbox {\footnote {\ }}\ \footnotetext { From Qu: {\color{red}#1}}}
\def\si#1 {\fbox {\footnote {\ }}\ \footnotetext { From Sihem: {\color{blue}#1}}}
\def\hyin#1 {}
\newcommand{\mqu}[1]{{\color{black}#1}}

\begin{document}

\title{On two-to-one mappings over finite fields}
\author{Sihem Mesnager{\thanks{Department of Mathematics, University of Paris VIII, 93526 Saint-Denis, France, University of Paris XIII, CNRS, LAGA UMR 7539, Sorbonne Paris Cit\'e, 93430 Villetaneuse,
France, and Telecom ParisTech 75013 Paris. Email: smesnager@univ-paris8.fr}}, Longjiang Qu{\thanks{Corresponding author. \newline National University of Defense Technology, Department of Mathematics, Changsha, China. E-mail: ljqu\_happy@hotmail.com. The research of L.J. Qu is supported by  the Nature Science Foundation of China (NSFC) under Grant 61722213, 11531002, 61572026, and the National Key R\&D Program of China (No. 2017YFB0802000).}}}

\date{\today}
\maketitle
\begin{abstract}
Two-to-one ($2$-to-$1$)  mappings over finite fields play an important role in symmetric cryptography. In particular they allow to design 
APN functions, bent functions and semi-bent functions. In this paper we provide a systematic study of two-to-one mappings that
are defined over finite fields. We characterize such mappings by means of the Walsh transforms. We also present several constructions, including an AGW-like criterion, constructions with the form of $x^rh(x^{(q-1)/d})$, those from permutation polynomials, from linear translators and from APN functions.
Then we present  $2$-to-$1$ polynomial mappings in classical classes of polynomials: linearized polynomials and monomials, low degree polynomials, Dickson polynomials and Muller-Cohen-Matthews polynomials, etc. Lastly, we show applications of $2$-to-$1$ mappings over finite fields for
constructions of bent Boolean and vectorial bent functions, semi-bent functions, planar functions and permutation polynomials. 
In all those respects, we shall review what is known and provide several new results.
\end{abstract}
{\bf Keywords: }Two-to-one mappings, permutation polynomials, AGW criterion, linear translators, symmetric cryptography.

\section{Introduction}

Permutation mappings (or $1$-to-$1$ mappings) over finite fields have been extensively studied for
their applications in cryptography, coding theory, combinatorial
design, etc. For recent advances on permutation polynomials over finite fields, we refer to the excellent survey \cite{Hou} and the references therein. For a detailed study of involutions over finite fields (in characteristic 2), we send the reader to \cite{CharpinMesnagerSarkar}. Two-to-one ($2$-to-$1$)  mappings are involved in several criteria in particular to design special important primitives in symmetric cryptography such as APN functions, bent functions and more general plateaued functions. Despite their importance, they have never been studied in the literature. The objective of this paper is to provide a systematic study of two-to-one mappings over finite fields including characterizations, criteria and methods for handling and designing such functions as well as effective constructions.

The paper is organized as follows.
Section \ref{Preliminaries} gives preliminaries and fixes the notation.  In  Section \ref{def}, we first present the definitions of $2$-to-$1$ mappings over finite fields as well as basic properties, and then provide a characterization of $2$-to-$1$ mappings by means of the Walsh transforms.
Section \ref{cons} is devoted to the constructions of  $2$-to-$1$ mappings. We shall present several constructions. First, an AGW-like criterion for $2$-to-$1$ mappings is given. Next, constructions of $2$-to-$1$ polynomial mappings with the form of $x^rh(x^{(q-1)/d})$ are provided. Furthermore, constructions of $2$-to-$1$ mappings from permutation polynomials, from linear translators and from APN functions respectively are given. In Section \ref{Pol} we present  $2$-to-$1$ polynomial mappings in classical classes of polynomials: linearized polynomials and monomials, low degree polynomials, Dickson polynomials and Muller-Cohen-Matthews polynomials, etc. In  Section \ref{appli}, we are interested in applications of $2$-to-$1$ mappings over finite fields for
constructions of  bent Boolean and vectorial bent functions, semi-bent functions, planar functions and permutation polynomials. 
{It should be noted that this section  is not only an application of the obtained results, but also a motivation to study 2-to-1 mappings.} In all those sections, we shall review what is known and provide several new results.

\section{Notation and Preliminaries}\label{Preliminaries}
\mqu{For  a set $S$,  $\# S$  will denote the cardinality of $S$. For any field $E$, $E^{\ast}=E \setminus\{0\}$.  Let $\mathbb{N}$,  $\mathbb{R}$ and $\mathbb{C}$  be respectively the set  of all natural,  real  and complex numbers.} Let $p$ be a prime number and $n$ be a positive integer. The finite field with $q:=p^n$ elements is denoted by $\F_q$ or $\F_{p^n}$, which  can be viewed as an $n$-dimensional vector space over $\F_{p}$, and it is denoted by $\F_p^n$. 
Denote by $\overline{\gf}_q$  the algebraic closure of $\gf_{q}$.  
The trace function $\Tr_{p^n/p}: \mathbb {F}_{p^n} \rightarrow \mathbb {F}_{p}$ is defined as
\begin{displaymath}
\Tr_{p^n/p}(x) =\sum_{i=0}^{ n-1}
  x^{p^{i}}=x+x^{p}+x^{p^2}+\cdots+x^{p^{n-1}},
  \end{displaymath}
 which is called \textit{the absolute trace} of $x\in \mathbb {F}_{p^n}$. 
 More general, the trace function $\Tr_{q^n/q}: \mathbb {F}_{q^n} \rightarrow \mathbb {F}_{q}$ is defined as 
 \begin{displaymath}
\Tr_{q^n/q}(x) =\sum_{i=0}^{ n-1}
  x^{q^{i}}=x+x^{q}+x^{q^2}+\cdots+x^{q^{n-1}}.
  \end{displaymath}
  
 A linearized polynomial (or additive polynomial), $L\in \mathbb{F}_q[x]$ is a polynomial of the shape $L(x)=\sum_{i=0}^n a_i x^{p^i}$. A polynomial $A\in \mathbb{F}_q[x]$ is called an affine polynomial if it equals to the summation of a linearized polynomial and a constant term.
  
 Let $f$ be a  function from $\mathbb {F}_p^n$  to $\mathbb {F}_{p}$. 
We can give a corresponding complex-valued function  $\chi_f$ from $\F_p^n$ to $\C $ defined as  $\chi_f(x)=\xi_p^{f(x)}$ for  all $ x\in\F_p^n$ where $ \xi_p=e^{(2\pi \sqrt {-1})/p}$ is a complex  primitive $p$-th root of unity.
 The Walsh transform  of $f$ is the Fourier transform $\widehat{\chi}_f$  from $\F_p^n$ to $\C $ of  $\chi_f$  defined as
$\widehat{\chi}_f(\omega)=\sum_{x\in  \mathbb {F}_p^n} {\xi_p}^{{f(x)}- \omega \cdot x}$ 
for all $\omega\in  \F_p^n$, where $``\cdot"$ denotes an inner product (for instance, the usual inner product) in $\F_p^n$. We can take $\omega \cdot x=\Tr_{p^n/p}(\omega x)$  if $\F_p^n$ is identified with $\F_{p^n}$. Note that if $p=2$ then $\xi_p=-1$ and a function from $\mathbb {F}_{2^n}$  to $\mathbb {F}_{2}$ is said to be a Boolean function.


\section{Definitions and a characterization of $2$-to-$1$ mappings over finite fields}\label{def}
{\subsection{Definitions of $2$-to-$1$ mappings}}
\mqu{Firstly, we give the definition of $2$-to-$1$ mappings over any finite set. 

\begin{definition}
	Let $A$ and $B$ be two finite sets, and let $f$ be a mapping from $A$ to $B$. Then $f$ is called a $2$-to-$1$ mapping if one of the following two cases hold:
	\begin{enumerate}
		\item $\sharp A$ is even, and for any $b\in B$, it has either $2$ or $0$ preimages of $f$;\newline
		\item $\sharp A$ is odd, and for all but one $b\in B$, it has either $2$ or $0$ preimages of $f$, and the exception element has exactly one preimage. 
	\end{enumerate}
\end{definition}	

Throughout this paper, we mainly focus on the mappings over finite fields. 
Let $\gf_{p^n}$ and $\gf_{p^m}$ be two finite fields of order  $p^n$ and $p^m$, respectively. Let $F$ be a mapping from $\gf_{p^n}$ to $\gf_{p^m}$.  
Then according to the above definition, if $p=2$, then $F$ is  a
\emph {$2$-to-$1$ mapping}  if and only if the equation $F(x)=a$ has either zero or two solutions in $\gf_{2^n}$ for any $a\in \gf_{2^m}$, or equivalently, $\#\{x\in \gf_{2^n} | F(x)=F(y)\}=2$ for all $y\in \gf_{2^n}$.
While for an odd prime $p$, a mapping $F: \gf_{p^n}\rightarrow \gf_{p^m}$ is $2$-to-$1$ if and only if all but one elements in the image set of $F$ have two preimages and the exceptional element has one preimage, or equivalently, there exists a unique $b\in\gf_{p^m}$ such that $\# F^{-1}(b)=1$ and $\# F^{-1}(a)\in\{0,2\}$, for all $a\in \gf_{p^m}\setminus\{b\}$. Without loss of generality, we can assume that the exceptional element of $b$ is $0$. Moreover, if its unique preimage is also the zero element, then we have the following remark.  }

	\begin{remark}
		Let $F  :\gf_{p^n}\rightarrow \gf_{p^m}$ with {$F(x)=0$ if and only if $x=0$}, where $p$ is odd. Then $F$ is a $2$-to-$1$ mapping  if and only if, $F(x)=a$ has either zero or two solutions in $\gf_{p^n}$ for any $a\in \gf_{p^m}^\ast$.
		\end{remark} 


In the end of this {subsection}, we calculate  the number of all $2$-to-$1$ mappings over $ \gf_{2^n}$.  It seems to be a huge number. 
\begin{proposition}
	Denote by  $N_n$  the number of all $2$-to-$1$ mappings $F: \gf_{2^n}\rightarrow \gf_{2^n}$. Then $$N_n=\frac{(2^n!)^2}{2^{2^{n-1}}(2^{n-1}!)^2}\approx \frac{2^{n\cdot 2^n+2^{n-1}+1}}{e^{2^n}}.$$ 
\end{proposition}

\begin{proof}
	Let $F$ be a $2$-to-$1$ mapping over $\gf_{2^n}$. Then the size of its image set  is $2^{n-1}$. For the first element of the image set, its preimage  have $\binom{2^n}{2}$ choices, while for the second element, it has $\binom{2^n-2}{2}$ choices, so on and so forth, the last element has $\binom{2}{2}$ choices. Hence we have
	  $$N_n=\binom{2^n}{2^{n-1}}\cdot 
	  \binom{2^n}{2}\cdot \binom{2^n-2}{2}\cdots\binom{2}{2}= \frac{(2^n!)^2}{2^{2^{n-1}}(2^{n-1}!)^2}.$$	
	  {Then the result follows from the well-known String formula.}
\end{proof}

It is well known that the number of all mappings  (resp. bijective mappings) from $ \gf_{2^n}$ to itself is $2^{n\cdot 2^n}$ (resp. $(2^n)!$) . Denote the latter number by $B_n$. Then we have 
	$$\frac{N_n}{B_n}=\frac{2^n!}{2^{2^{n-1}}(2^{n-1}!)^2}\approx \frac{2^{2^{n-1}}}{\sqrt{\pi 2^{n-1}}}.$$
	
	\mqu{We list the ratio of these two numbers for $1\leq n\leq 8$ in the following table.  The values are rounded to three significant figures.
\begin{center}
	\begin{tabular}{|c|c|c|c|c|c|c|c|c|c|c|c|} \hline
		\label{table-N1}
		$n$ & $1$ & $2$ & $3$ & $4$ & $5$  & $6$ & $7$ & $8$  \\ \hline
			${N_n}/{B_n}$ & $1.00$ & $1.50$ & $4.38$ & $50.3$ & $9.17\cdot 10^3$  & $4.27\cdot 10^{8}$& $1.30\cdot 10^{18}$ & $1.70\cdot 10^{37}$  \\ \hline
	\end{tabular}
\end{center}

		It seems from the above table} that the number of all $2$-to-$1$ mappings over $ \gf_{2^n}$
		 is much greater than that of all bijective mappings over $ \gf_{2^n}$. 

{\subsection{A characterization of $2$-to-$1$ mappings over  $\gf_{2^n}$ by means of the Walsh transforms}\label{char}}

In this {subsection} we present a characterization of $2$-to-$1$ mappings over $\gf_{2^n}$ by means of the Walsh transforms. The main idea goes back to Carlet \cite{Carlet2017} who has characterized the differential uniformity of vectorial functions by the Walsh transform. Let $F(x)$ be a polynomial over $\gf_{2^n}$. Recall that $F(x)$ is $2$-to-$1$ if and only if, for every $b$ in $\gf_{2^n}$, the equation $F(x) =b$ has 0 or 2 solutions. Let $F: {\Bbb F}_{2^n}\rightarrow {\Bbb F}_{2^n}$ be a vectorial Boolean function. The Walsh transform of $F$ at
$(u,v)\in {\Bbb F}_{2^n}\times{\Bbb F}_{2^n}$ equals  by definition the Walsh transform
of the so-called component function $\Tr_{2^n/2}(v F(x))$ at $u$, that is:
$$W_F(u,v):=\sum_{x\in {\Bbb F}_{2^n}}(-1)^{\Tr_{2^n/2}(v F(x))+\Tr_{2^n/2}(ux)}.$$

\mqu{
Let $\phi(x)=\sum_{j\ge0}A_jx^j$ be any polynomial over $\mathbb{R}$ such that $\phi(x)=0$ for $x=0, 2$ and $\phi(x)>0$ for every $x\in \mathbb{N} \setminus \{0, 2\}$. Hence for any $F$ and $b\in\gf_{2^n}$, we have 
	$$\sum_{j\ge0}A_j (\#\{x\in \gf_{2^n}: F(x)+b=0\}  )^j\ge0, $$
	and $F$ is a two-to-one function if and only if this inequality is an equality for any  $b\in\gf_{2^n}$. Furthermore, for any $F$, we have 
	$$\sum_{j\ge0}A_j \sum_{b\in\gf_{2}^n}  (\#\{x\in \gf_{2^n}: F(x)+b=0\}  )^j \ge 0,  $$
	and $F$ is $2$-to-$1$ if and only if this inequality is an equality. 

We shall now characterize this condition by means of the Walsh transform. We
have:
\begin{equation*}
\#\{x\in \gf_{2^n}: F(x)+b=0\}=2^{-n}\sum_{x\in\gf_{2^n}, v\in\gf_{2^n}}(-1)^{\Tr_{2^n/2}(v(F(x)+b))}, 
\end{equation*}
and therefore, for $j\geq 1$:
\begin{eqnarray*}
	&& \sum_{b\in \gf_{2^n}}\left( \#\{x\in \gf_{2^n}: F(x)+b=0\} \right)^j\\
	&=& 2^{-jn}\sum_{b\in \gf_{2^n}}\sum_{\begin{array}{c}
			x_1, \cdots, x_j\in\gf_{2^n},\\
			v_1,\cdots,v_j\in\gf_{2^n}
	\end{array}} (-1)^{\sum_{i=1}^{j}\Tr_{2^n/2}(v_i(F(x_i)+b))}\\
	&=& 2^{-jn}\sum_{\begin{array}{c}
			x_1, \cdots, x_j\in\gf_{2^n},\\
			v_1,\cdots,v_j\in\gf_{2^n}
	\end{array}} (-1)^{\sum_{i=1}^{j}\Tr_{2^n/2}(v_iF(x_i))}\sum_{b\in \gf_{2^n}}(-1)^{\Tr_{2^n/2}(b\sum_{i=1}^{j}v_i)}\\
	&=& 2^{-(j-1)n} \sum_{\begin{array}{c}
			v_1,\cdots,v_j\in\gf_{2^n}\\
			\sum_{i=1}^jv_i=0
	\end{array}} \prod_{i=1}^{j} W_F(0, v_i).
\end{eqnarray*}
Hence we have the following characterization of $2$-to-$1$ mappings over  $\gf_{2^n}$ by the Walsh transform.
\begin{theorem}
	Let $F: {\Bbb F}_{2^n}\rightarrow {\Bbb F}_{2^n}$ be a vectorial Boolean
	function. Then
		$$A_0  + \sum_{j\ge1}A_j   2^{-(j-1)n} \sum_{\begin{array}{c}
			v_1,\cdots,v_j\in\gf_{2^n}\\
			\sum_{i=1}^jv_i=0
			\end{array}} \prod_{i=1}^{j} W_F(0, v_i) \ge 0,  $$
	and $F$ is $2$-to-$1$ if and only if this inequality is an equality. 
\end{theorem}

 Now, let us consider the polynomial over $\mathbb{R}$ equal to $X(X-2)^2 = X^3-4X^2+4X$.
It takes value 0 when $X$ equals 0 or 2 and takes strictly positive value when $X$ is in $\mathbb{N}\setminus \{0,2\}$. We have then the following corollary. 
%
%

\begin{corollary}
	Let $F: {\Bbb F}_{2^n}\rightarrow {\Bbb F}_{2^n}$ be a vectorial Boolean
	function. Then
	\begin{equation*}
	2^{-2n}\sum_{v_1,v_2\in\gf_{2^n}}  W_F(0, v_1)W_F(0, v_2) W_F(0, v_1+v_2) - 2^{2-n}\sum_{v\in\gf_{2^n}} W_F(0, v)^2 +2^{n+2}\geq 0, 
	\end{equation*}
	and this inequality is an equality if and only if $F$ is $2$-to-$1$.
\end{corollary}
}

\section{Constructions of  $2$-to-$1$ mappings}\label{cons}
In this section, we present different methods to construct $2$-to-$1$ mappings over finite fields.

\subsection{AGW-like criterion for $2$-to-$1$ mappings}

The criterion, discovered by Akbary, Ghioca
and Wang \cite{AGW2011}, is a simple and effective method that establishes the permutation property of a mapping $\gf_q\rightarrow \gf_q$ through a commutative diagram. 
The significance of the AGW criterion resides in the fact that it not only provides
a unified interpretation for many previous constructions of permutations polynomials but also facilitates numerous new discoveries. In this subsection we will generalize AGW criterion to construct $2$-to-$1$ mappings over finite fields. 

\mqu{We give a brief description of this subsection for the readers' convenience. 
First, the AGW criterion is generalized to construct $2$-to-$1$ mappings over finite sets (Proposition \ref{AGW-Gen}). Second, three general constructions (Theorem \ref{field_gen}, Theorem \ref{AGW_3L}, and  Proposition \ref{prop_psi}) are  given by applying this generalized AGW criterion. Then several explicit $2$-to-$1$ polynomials over finite fields are constructed from Proposition \ref{prop_psi}, and most of the constructions are divided into two cases. 
}

{
\begin{proposition}\label{AGW-Gen}
	Let $A$ be a finite set, $S, \bar{S}$ be two finite sets such that $\sharp S=\sharp \bar{S}$. Let  $f, g, \lambda, \bar{\lambda}$ be four mappings defined as the following diagram such that $\bar{\lambda}\circ f= g\circ \lambda$.  If $g$ is bijective from $S$ to $\bar{S}$,  $f|_{\lambda^{-1}(s)}$ is $2$-to-$1$ for any $s\in S$, and there is at most one $s\in S$ such that $\sharp {\lambda^{-1}(s)}$ is odd, then $f$ is a $2$-to-$1$ mapping over $A$.
	
	\begin{equation*}
	\xymatrix{
		A \ar[rr]^{f}\ar[d]_{\lambda} &   & A \ar[d]^{\bar{\lambda}} \\
		S \ar[rr]^{g} &  & \bar{S} }
	\end{equation*} 
\end{proposition}

\begin{proof} 
Let $b\in A$. Assume that there exists an element $a$ in $A$ such that $f(a)=b$. 
Let $\bar{s}=\bar{\lambda}(b)$. Then 
$$\bar{s}=\bar{\lambda}(b) = \bar{\lambda}\circ f(a) = g\circ \lambda(a). $$
Since $g$ is bijective from $S$ to $\bar{S}$, there exists a unique element $s\in S$ such that $g(s)=\bar{s}$. Hence $\lambda(a)=s$. 
 If $\sharp {\lambda^{-1}(s)}$ is even, then $f(x)=b$ has exactly two solutions in $A$ (one is $a$) since 
$f|_{\lambda^{-1}(s)}$ is $2$-to-$1$ for any $s\in S$. If $\sharp {\lambda^{-1}(s)}$ is odd, then with one exception, $b$ has exactly two preimages of $f$ in $A$. Further, 
since  at most one of $\sharp {\lambda^{-1}(s)}$ is odd for all $s\in S$,
we know that  $f$ is a $2$-to-$1$ mapping over $A$.  
\end{proof}

\begin{remark}
 If $g$ is $2$-to-$1$ from $S$ to $\bar{S}$, and $f|_{\lambda^{-1}(s)}$ is injective for any $s\in S$, then one can only deduce that for any $b\in A$, it has at most two preimages. Similarly, let $b\in A$ and assume that  there exists an element $a$ in $A$ such that $f(a)=b$. 		
		Let $\bar{s}=\bar{\lambda}(b)$. Then 
		$$\bar{s}=\bar{\lambda}(b) = \bar{\lambda}\circ f(a) = g\circ \lambda(a). $$
		Since $g$ is $2$-to-$1$, there exist exactly two elements $s_1, s_2$ in $S$ such that $g(s_1)=g(s_2)=\bar{s}$ with at most one exception. Hence $\lambda(a)=s_1$ or $\lambda(a)=s_2$. 
		Then it follows from the assumptions that $f|_{\lambda^{-1}(s)}$ is $1$-to-$1$ for any $s\in S$ that
		there exist at most  two elements $a_1(=a), a_2$ in $A$ such that $f(a_1)=f(a_2)=b$. It seems not easy to add a condition such that $f$ is a $2$-to-$1$ mapping over $A$. We leave this problem to interested readers. 
\end{remark}

By applying Proposition \ref{AGW-Gen}, we can give the following two general constructions. 

\begin{theorem}\label{field_gen}
	Consider any polynomial $g \in \gf_{q^n}[x]$,  any additive polynomials $\phi, \psi  \in \gf_{q^n}[x]$, any $\gf_q$-linear
	polynomial $\bar{\psi}  \in \gf_{q^n}[x]$ satisfying $\phi\circ \psi = \bar{\psi}\circ \phi $, and any polynomial  $h \in \gf_{q^n}[x]$
	such that $h(\psi(\gf_{q^n})) \subseteq \gf_q^\ast$.
	Let $$f(x)=h(\psi(x))\phi(x)+g(\psi(x))$$ and $$\bar{f}(x)=h(x)\phi(x)+\bar{\psi}(g(x)).$$ 
If	$\bar{f}$ is bijective from  $\psi(\gf_{q^n})$ to $\bar{\psi}(\gf_{q^n})$, $f|_{\psi^{-1}(s)}$ is $2$-to-$1$ for any $s\in \psi(\gf_{q^n})$, and 
there is at most one $s\in  \psi(\gf_{q^n})$ such that $\sharp { \psi^{-1}(s)}$ is odd,
then $f$ is a $2$-to-$1$ mapping over $\gf_{q^n}$.
\end{theorem}
\begin{proof}
We have	\begin{eqnarray*}
		\bar{\psi}\circ f(x) & = & \bar{\psi}(h(\psi(x))\phi(x))  +  \bar{\psi}(g(\psi(x)))\\
		&=& h(\psi(x))\phi(\psi(x)) +  \bar{\psi}(g(\psi(x)))\\
		&=&	 \bar{f}\circ \psi(x), 
	\end{eqnarray*}
	the second equality holds since $h(\psi(\gf_{q^n})) \subseteq \gf_q^\ast$, $\bar{\psi}$ is $\gf_q$-linear and $\phi\circ \psi = \bar{\psi}\circ \phi $. Hence we get the following commutative diagram:
	\begin{equation*}
	\xymatrix{
		\gf_{q^n} \ar[rr]^{f}\ar[d]_{\psi} &   & \gf_{q^n} \ar[d]^{\bar{\psi}} \\
		\psi(\gf_{q^n}) \ar[rr]^{\bar{f}} &  & \bar{\psi}(\gf_{q^n}) }
	\end{equation*} 
	Then the result follows directly from Proposition \ref{AGW-Gen}. 
\end{proof}

\begin{theorem}\label{AGW_3L}
	Let $q$ be an even prime power, let $n$ be a positive integer, and let $L_1, L_2, L_3$ be $\gf_q$-linear polynomials
	over $\gf_q$ seen as endomorphisms of  \mqu{ the $\gf_q$-module } $\gf_{q^n}$. Let $g(x)\in \gf_{q^n}[x]$ be such that $g(L_3(\gf_{q^n}))\subseteq \gf_q$. Assume
	$$f(x)=L_1(x)+L_2(x)g(L_3(x))$$
	and 
	$$\bar{f} (x) = L_1(x) + L_2(x)g(x). $$
	For any \mqu{$y\in L_3(\gf_{q^n})$}, let $F_y(x) := L_1(x) + L_2(x)g(y)$. 
	If $\ker(F_ y) \cap \ker(L_3) = \{0, c_y\}$, for any  \mqu{$y\in L_3(\gf_{q^n})$}, where 	$c_y$ is a nonzero element of $\gf_{q^n}$,  and $\bar{f}$ is a permutation over $L_3(\gf_{q^n})$, then $f$ is $2$-to-$1$ over $\gf_{q^n}$.	
\end{theorem}
\begin{equation*}
\xymatrix{
	\gf_{q^n} \ar[rr]^{f}\ar[d]_{L_3} &   & \gf_{q^n} \ar[d]^{L_3} \\
	L_3(\gf_{q^n}) \ar[rr]^{\bar{f}} &  & L_3(\gf_{q^n}) }
\end{equation*}
\begin{proof} We apply Proposition \ref{AGW-Gen} with $A = \gf_{q^n}$, 
	$f(x)=L_1(x)+L_2(x)g(L_3(x))$, $S =\bar{S} = L_3(\gf_{q^n})$, $\lambda = \bar{\lambda} = L_3$
	and $\bar{f} (x) := L_1(x) + L_2(x)g(x)$.
	Since $g(L_3(\gf_{q^n}))\subseteq \gf_q$, and $L_1, L_2$ and $L_3$ are $\gf_q$-linear polynomials over $\gf_q$,  one can easily verified that $\lambda\circ f = \bar{f}\circ \lambda$. 
	For any \mqu{$y\in L_3(\gf_{q^n})$}, $ f|_{\lambda^{-1}(y)}= L_1(x) + L_2(x)g(y) = F_y$ is linearized. It is $2$-to-$1$ over $L_3^{-1}(y)$ if and only if $\dim (\ker(F_ y) \cap \ker(L_3))=1$.  Hence the result follows from Proposition \ref{AGW-Gen}. 
\end{proof}

\mqu{
The above two constructions are quite general and can be used to construct more explicit $2$-to-$1$ polynomials.  Due to the space limit, we will only take the first one as an example and give several explicit  constructions. The interested readers are cordially invited to apply the second one to construct more $2$-to-$1$ polynomials.

The following proposition follows from Theorem \ref{field_gen}, and is the foundation of later constructions in this subsection.   }

\begin{proposition}\label{prop_psi}
	Let $q=2^m$, $\phi(x)$ and $\psi(x)$ be two $\gf_q$-linear polynomials over $\gf_q$ seen as endomorphisms of \mqu{ the $\gf_q$-module } $\gf_{q^n}$, and
	let $g, h\in \gf_{q^n}[x]$ such that $h(\psi(\gf_{q^n})) \subseteq \gf_q^\ast.$ If $\ker(\phi) \cap \ker(\psi) = \{0, c\}$ for some $c\in\gf_{q^n}^\ast$, and $\bar{f}(x)=h(x)\phi(x) + \psi(g(x))$ permutes $\psi(\gf_{q^n})$, then
	$$f(x)=h(\psi(x))\phi(x) + g(\psi(x))$$
	is $2$-to-$1$ over $\gf_{q^n}$.
\end{proposition}
\begin{proof}
	In Theorem \ref{field_gen}, let $\bar{\psi}=\psi$, then  $\phi\circ \psi = {\psi}\circ \phi $ since both $\psi$ and $\phi$ are $\gf_q$-linear polynomials \mqu{over $\gf_q$}. Further, $f|_{\psi^{-1}(s)}$ is  $2$-to-$1$ for any $s\in \psi(\gf_{q^n})$ since  $\dim_{\gf_2}(\ker(\phi) \cap \ker(\psi)) = 1$. The result then follows from Theorem \ref{field_gen}. 
\end{proof}
 
By applying Proposition \ref{prop_psi}, we have the following theorem. 
\begin{theorem}
	Let $q=2^m$, $a\in\gf_q$, and let $b\in \gf_{q^n}$. Let $P(x)$ and $L(x)$ be $\gf_q$-linear polynomials
	over $\gf_q$. Let $H(x)\in \gf_{q^n}[x]$ be such that $H(L(\gf_{q^n})) \subseteq \gf_q\setminus \{-a\}$. Let 
	$$f(x)=aP(x) + (P(x)+b)H(L(x))$$
	and 			$$\bar{f}(x)=aP(x) + (P(x)+L(b))H(x). $$
If $\ker(P) \cap \ker(L) = \{0, c\}$ for some  $c\in\gf_{q^n}^\ast$,  and $\bar{f}$ permutes $L(\gf_{q^n})$,
	then $f$ is $2$-to-$1$ over $\gf_{q^n}$.
\end{theorem}
{\bfseries Proof. }
In Proposition \ref{prop_psi}, we let $h(x) = a + H(x)$, $\phi(x) = P(x)$, $\psi(x) = L(x)$ and $g(x) = b\cdot H(x)$.
For any $x \in L(\gf_{q^n})$, since $H(x)\in \gf_q$ and $L$ is a $\gf_q$-linear polynomial, we obtain
$$L(b) \cdot H(x) = L(b\cdot H(x)) = \psi(g(x)), $$
and thus $$\bar{f}(x) = (a + H(x))P(x) + L(b)H(x) = h(x)\phi(x) + \psi(g(x)),$$ as in
Proposition \ref{prop_psi}.
$\hfill\Box$ 

Next we study in detail some of the consequences of Proposition \ref{prop_psi} (or alternatively of Theorem \ref{field_gen} when $\psi= \bar{\psi}$ ) for two specific choices of $\gf_q$-linear polynomials. First we consider the case
$\psi(x) = \Tr_{q^n/q}(x)$ and next we study the case $\psi(x) = x^q - x$.

{\bfseries Case 1. }$\psi(x)=\bar{\psi}(x)=\Tr_{q^n/q}(x) = x + x^q + \cdots  + x^{q^{n-1}}.$

The first result in this case follows directly from Proposition \ref{prop_psi}. 
\begin{proposition}\label{pro_Tr}
	Let  $q=2^m$, $\phi(x)$ be a $\gf_q$-linear polynomial over $\gf_q$ seen as an endomorphism of  \mqu{ the $\gf_q$-module } $\gf_{q^n}$ and $\Tr_{q^n/q}(x)$ be the trace function from $\gf_{q^n}$ to $\gf_q$.  Let $g, h\in \gf_{q^n}[x]$ \mqu{be such that} $h(\gf_{q}) \subseteq \gf_q^\ast.$ Assume 
	$$f(x)=h(\Tr_{q^n/q}(x))\phi(x) + g(\Tr_{q^n/q}(x))$$
	and $$\bar{f}(x) = h(x)\phi(x) + \Tr_{q^n/q}(g(x)).$$
	If $\ker(\phi) \cap \ker(\Tr_{q^n/q}(x)) = \{0, c\}$ for some  $c\in\gf_{q^n}^\ast$, and  $\bar{f}$ permutes $\gf_q$, then $f$ is $2$-to-$1$ over $\gf_{q^n}$.
\end{proposition}

By applying Proposition \ref{pro_Tr}, we get the following construction. 
\begin{theorem}
	Let $q=2^m$, $\phi$ be a $\gf_{q}$-linear polynomial over $\gf_q$, let $g(x) \in \gf_{q^n}[x]$, and let $h(x)\in \gf_{q^n}[x]$ such that
	$h(\gf_{q}) \subseteq \gf_q^\ast$. Assume $f (x) = h(\Tr_{q^n/q}(x))\phi(x) + g(\Tr_{q^n/q}(x))^q - g(\Tr_{q^n/q}(x))$. If $\ker(\phi) \cap \ker(\Tr_{q^n/q}(x)) = \{0, c\}$ for some  $c\in\gf_{q^n}^\ast$, and  $h(x)\phi(x)$  permutes $\gf_q$, then 
	$f$ is $2$-to-$1$ over $\gf_{q^n}$.
\end{theorem}

{\bfseries Case 2. } $\psi(x)=\bar{\psi}(x)=x^q - x.$


Similarly, we have the following two results. 

\begin{proposition}\label{prop_xq}
	Let $q=2^m$, $\phi(x)$ be a $\gf_q$-linear polynomials over $\gf_q$ seen as an endomorphism of \mqu{ the $\gf_q$-module } $\gf_{q^n}$.  Let $g, h\in \gf_{q^n}[x]$ \mqu{be such that} $h(x^q-x) \subseteq \gf_q^\ast$  for 
	all $x\in \gf_{q^n}$.  Assume 
	$$f(x)=h(x^q-x)\phi(x) + g(x^q-x)$$
	and $$\bar{f}(x) = h(x)\phi(x) + g(x)^q-g(x).$$
If $\phi(x)$ is $2$-to-$1$ over $\gf_q$ and $\bar{f}$ permutes over $S = \{a^q - a| a\in  \gf_{q^n}\}$,	then $f$ is $2$-to-$1$ over $\gf_{q^n}$.
\end{proposition}
\begin{proof}
	In Theorem \ref{field_gen}, let $\psi =\bar{\psi} =x^q-x$. For any $s\in\psi(\gf_{q^n})$, $f|_{\psi^{-1}(s)}=h(s)\phi(x)+g(s)$ is $2$-to-$1$  if and only if $\phi$ is $2$-to-$1$ over $\gf_q$ since $h(s)\neq 0$. Hence the result follows. 	
\end{proof}

\begin{theorem}
	Let  $q=2^m$, $\phi(x)$ be a $\gf_q$-linear polynomials over $\gf_q$ seen as an endomorphism of \mqu{ the $\gf_q$-module } $\gf_{q^n}$.  Let $u, h\in \gf_{q^n}[x]$ \mqu{be such that} $h(x^q-x) \subseteq \gf_q^\ast$  for 	all $x\in \gf_{q^n}$.  Assume 
	$$f_1(x)=h(x^q-x)\phi(x) + \Tr_{q^n/q}(u(x^q-x))$$ and 
	$$ f_2(x)=h(x^q-x)\phi(x) + u(x^q-x)^{(q^n-1)/(q-1)}. $$
	If $\phi(x)$ is $2$-to-$1$ over $\gf_q$ and $h(x)\phi(x)$ permutes over $S = \{a^q - a| a\in  \gf_{q^n}\}$, then both $f_1$ and $f_2$ are $2$-to-$1$ over $\gf_{q^n}$.
\end{theorem}
\begin{proof}
	We only prove for $f_1$ as the other case can be proved similarly. 
	In Proposition \ref{prop_xq},  let $g(x)=\Tr_{q^n/q}(u(x))$. Then 
	$$\bar{f}_1(x) = h(x)\phi(x)+\psi(g(x)) = h(x)\phi(x)$$ 
	since $\psi(g(x))\equiv 0$.  Hence the result follows. 	
\end{proof}

\subsection{$2$-to-$1$ polynomial mappings with the form of $x^rh(x^{(q-1)/d})$}

In this subsection, we construct two-to-one polynomial mappings with the form of $x^rh(x^{(q-1)/d})$. We need $q-1$ to be even. Hence it is assumed that $q$ is odd throughout this subsection.

\begin{proposition}\label{prop_xrh}
	Let $q$ be an odd prime power, $r, d$ be positive integers such that $d|q-1$. Let $f(x)=x^rh(x^{(q-1)/d})$, where $h\in \gf_q[x]$ such that $h(x)\neq 0$ if $x\neq 0$, and let $\lambda(x)=x^{(q-1)/d}$ and \mqu{$\mu_d=\{x\in {\gf_q}: x^d=1\}$}. 	\mqu{Let $g(x)=x^rh(x)^{(q-1)/d}$. } If $g$ is $1$-to-$1$ from $\mu_d$ to $\mu_d$ and $\gcd(r, \frac{q-1}{d})=2$,
	then $f$ is a $2$-to-$1$ mapping over $\gf_q$.
	\mqu{
	\begin{equation*}
	\xymatrix{
	\gf_q^\ast \ar[rrr]^{\ \ \ f(x)=x^rh(x^{(q-1)/d})\ \ \ }\ar[d]_{x^{(q-1)/d}} &   & &  \gf_q^\ast \ar[d]^{x^{(q-1)/d}} \\
		\mu_{d} \ar[rrr]^{g(x)=x^rh(x)^{(q-1)/d}} &  & & 	\mu_{d} }
	\end{equation*} }
\end{proposition}

\begin{proof} 
 Since $\gcd(r, \frac{q-1}{d})=2$, we know that $f|_{\lambda^{-1}(s)}=x^rh(s)$ is $2$-to-$1$ for any $s\in \mu_d$. Then the  result follows directly from the fact that $\lambda \circ f=g\circ \lambda$ and Proposition \ref{AGW-Gen}.
\end{proof}

Then we have the following result.
\begin{corollary}\label{rqd_n}
	 Suppose that there exists $n \ge 0$ such that
	$h(x)^{(q-1)/d}  = x ^ n$ for all $x\in \mu_d$. If $\gcd(r + n, d)=1$ and $ \gcd(r, \frac{q-1}{d}) = 2$, then
	 $x^{r} h(x^{(q-1)/d}  )$ is $2$-to-$1$ over $\gf_q$.
\end{corollary}

\begin{theorem}\label{x_be_ga}
	Suppose that $q = q_0^m$, 	where $q_0\equiv 1(\mod d)$ and $d | m$,  $h \in \gf_{q_{0}} [X]$ has no roots in $ \mu_d$, $\gcd(r, d)=1$  and $\gcd(r, (q-1)/d) = 2$. Then $x^r h(x^{(q-1)/d}  )$ is $2$-to-$1$ over $\gf_q$.  
\end{theorem}
\begin{proof} 
For any $x\in \mu_d\subseteq \gf_{q_{0}}^\ast$, we have
\begin{eqnarray*}
	h(x)^{(q-1)/d} &=& h(x)^{\frac{q-1}{q_0-1}\cdot \frac{q_0-1}{d} }\\
	&=& \left(h(x)^{1+q_0+q_0^2+\cdots+q_0^{m-1}}\right)^{ \frac{q_0-1}{d}  }\\
	&=& \left(h(x)\cdot h(x^{q_0})\cdot \cdots \cdot h(x^{q_0^{m-1}}) \right)^{ \frac{q_0-1}{d}  }\\
	&=& \left(h(x)^m\right)^{ \frac{q_0-1}{d}  }=\left(h(x)^{q_0-1}\right)^{ \frac{m}{d}  }=1.	
\end{eqnarray*}
Hence  it is the case $n=0$ in Corollary \ref{rqd_n}. Then the result
follows. 
\end{proof}

}
	
\subsection{Constructions of $2$-to-$1$ mapping from permutation polynomials}

\begin{proposition}
	Let $G:\gf_{2^n}\rightarrow \gf_{2^n}$ be a permutation polynomial, and let $\gf_{2^n}=S_1\cup S_2$ be a disjoint decomposition of $\gf_{2^n}$, where $\#S_1=\#S_2=2^{n-1}$. Define $\phi$ be a bijective mapping from $S_2$ to $S_1$. Let 
	\begin{equation*}
	F(x) = \left\{\begin{array}{ll}
	G(x), & \text{if}\  x\in S_1;\\
	G(\phi(x)), & \text{if}\  x\in S_2. 
	\end{array}\right. 
	\end{equation*}
	Then $F$ is a $2$-to-$1$ mapping over $\gf_{2^n}$. 
\end{proposition}

Conversely, any $2$-to-$1$ mapping can be constructed by this method. 

The following corollary follows directly from the above proposition. 
\begin{corollary}
	Let $G:\gf_{2^n}\rightarrow \gf_{2^n}$ be a permutation polynomial, and let $S$ be a $\gf_2$-linear subspace of $\gf_{2^n}$ with dimension $n-1$, $\gamma\in \gf_{2^n}\setminus S$. 
	Define 
	\begin{equation*}
	F_1(x) = \left\{\begin{array}{ll}
	G(x), & \text{if}\  x\in S;\\
	G(x+\gamma), & \text{otherwise,} 
	\end{array}\right. 
	\end{equation*}
	and 
	\begin{equation*}
	F_2(x) = \left\{\begin{array}{ll}
	G(x+\gamma), & \text{if}\  x\in S;\\
	G(x), & \text{otherwise.} 
	\end{array}\right. 
	\end{equation*}
	Then both $F_1$ and $F_2$ are $2$-to-$1$ mappings over $\gf_{2^n}$.
\end{corollary}

In the above corollary, let $S=\{x\in \gf_q : \Tr_{2^n/2} (x)=0\}$, we have 
\begin{proposition}
	Let $G:\gf_{2^n}\rightarrow \gf_{2^n}$ be a permutation polynomial, and let  $\gamma\in \gf_{2^n}$ such that $\Tr_{2^n/2}(\gamma)=1$. 
	Define 
	\begin{equation*}
	F_1(x) = \left\{\begin{array}{ll}
	G(x), & \text{if}\  \Tr_{2^n/2} (x)=0;\\
	G(x+\gamma), & \text{otherwise,} 
	\end{array}\right. 
	\end{equation*}
	and 
	\begin{equation*}
	F_2(x) = \left\{\begin{array}{ll}
	G(x+\gamma), & \text{if}\ \Tr_{2^n/2} (x)=0;\\
	G(x), & \text{otherwise.} 
	\end{array}\right. 
	\end{equation*}
	Then both $F_1$ and $F_2$ are $2$-to-$1$ mappings over $\gf_{2^n}$.
\end{proposition}

New $2$-to-$1$ mappings can also be constructed from the composition of permutation polynomials and known  $2$-to-$1$ mappings.

\begin{proposition}
	Let $G:\gf_{2^n}\rightarrow \gf_{2^n}$ be a permutation polynomial, and let $H:\gf_{2^n}\rightarrow \gf_{2^n}$ be a $2$-to-$1$ mapping.
	Then both $F_1(x)=H(G(x))$ and $F_2(x)=G(H(x))$ are $2$-to-$1$ mappings over $\gf_{2^n}$. 
\end{proposition}

Hence one can use any $2$-to-$1$ polynomial and any permutation polynomial to 
produce new $2$-to-$1$ polynomials. It should be noted that in the above proposition we can composize many permutation polynomials with one $2$-to-$1$ mapping with any order. Particularly, we have the following corollary. 

\begin{corollary}
	Let $G:\gf_{2^n}\rightarrow \gf_{2^n}$ be a permutation polynomial, and let $L(x):\gf_{2^n}\rightarrow \gf_{2^n}$ be a linearized polynomial with $\dim \ker L=1$.
Then both $F_1(x)=L(G(x))$ and $F_2(x)=G(L(x))$ are $2$-to-$1$ mappings over $\gf_{2^n}$.
\end{corollary}

\subsection{Constructions of $2$-to-$1$ mappings from linear translators}
We recall the definitions of linear translator and linear structure.
\begin{definition}\label{de:tr}
	Let $n=rk$, $1\leq k\leq n$. Let $f$ be a function from $\mathbb{F}_{p^n}$ to $\mathbb{F}_{p^k}$,
	$\gamma\in\mathbb{F}_{p^n}^*$
	and $b$ be fixed in  $\mathbb{F}_{p^k}$.
	Then $\gamma$ is a $b$-{\it linear translator} of $f$ if
	$f(x+u\gamma) - f(x) = ub$ for all $x\in\mathbb{F}_{p^n}$ and $u\in\mathbb{F}_{p^k}$. In particular, if $k=1$ then  $\gamma$ is usually called  a
	$b$-{\it linear structure} of the function $f$ (where $b\in \mathbb{F}_{p}$), that is $f(x+\gamma)-f(x)=b$ for all $x\in\mathbb{F}_{p^n}$.
\end{definition}

%

\begin{proposition}\cite{CharpinKyureghyanFFA}
	Let  $G$ be a polynomial in $\mathbb{F}_{2^n}[x]$, $F$ be a permutation on $\mathbb{F}_{2^n}$ and $\gamma\in \mathbb{F}_{2^n}$ be a $1$-linear structure of $\Tr_ {2^n/2}(G(x))$. Then $F(x)+\gamma \Tr_ {2^n/2}(G(F(x)))$ is a $2$-to-$1$ mapping over $\mathbb{F}_{2^n}$.
\end{proposition}


The following result can be derived from \cite{Mesnager-Oz}. For making the paper self-contained, we include its proof.
\begin{proposition}
	Let {$\gamma, \delta$ be two distinct elements in $\mathbb{F}_{2^m}^\ast$,  and let} $f, g$ be two Boolean functions defined over $\mathbb{F}_{2^m}$. The mapping  {$y\mapsto y+\gamma f(y)+\delta g(y)$} is  $2$-to-$1$ on $\GF m$ if one of the following conditions holds:
\begin{enumerate}
\item $\gamma$, $\delta$ are $1$-linear structures of $f$ and $\gamma$ is a $0$-linear structure of $g$,
\item {$\gamma$} is a $1$-linear structure of $f$ and $\gamma$, $\delta$ are $0$-linear structures of $g$,
\item $\gamma$, $\delta$ are $0$-linear structures of $f$ and $\delta$ is a $1$-linear structure of $g$,
\item $\delta$ is a $0$-linear structure of $f$ and $\gamma$, $\delta$ are $1$-linear structures of $g$,
\item $\gamma$ is a $0$-linear structure of $f$, $\delta$ is a $1$-linear structure of $f$ and $\gamma+\delta$ is a $1$-linear structure of $g$,
\item $\gamma$ is a $1$-linear structure of $g$, $\delta$ is a $0$-linear structure of $g$ and $\gamma+\delta$ is a $1$-linear structure of $f$.
\end{enumerate}
\end{proposition}

\begin{proof}
We give the proof for Case 1 only since the proofs for other cases are similar. Now, we need to show that $\rho(y) : y \mapsto y+\gamma f(y)+\delta g(y)$ is $2$-to-$1$. Let $\rho(y)=a$ for some $a \in \GF m$. Then, $y \in \left\{a,a+\gamma,a+\delta,a+\gamma+\delta\right\}$. As $\gamma$ is a $1$-linear structure of $f$ and $0$-linear structure of $g$, we have $\rho(a)=\rho(a+\gamma)$ and $\rho(a+\delta)=\rho(a+\gamma+\delta)$. Moreover, $\rho(a+\delta)=a+\delta+\gamma f(a+\delta)+ \delta g(a+ \delta)= a+\delta +\gamma + \gamma f(a) + \delta g(a+\delta)$ where we use that $\delta$ is a $1$-linear structure of $f$. We observe that $\rho(a)=a+\gamma f(a) + \delta g(a) \neq \rho(a+\delta)$. Indeed, \mqu{if} the equality holds, then $\gamma + \delta + \delta \big(g(a)+g(a+\delta)\big)=0$. This is a contradiction as $\gamma \neq \delta$ and $\gamma \neq 0$. This implies that $\rho^{-1}(a)=\left\{a,a+\gamma\right\}$ or $\rho^{-1}(a)=\left\{a+\delta,a+\gamma+\delta\right\}$ which shows that $\rho$ is $2$-to-$1$.
\end{proof}

\begin{proposition}
	Let $L : \mathbb{F}_{2^n}\rightarrow \mathbb{F}_{2^n}$ be a $\GF {}$-linear permutation of $\GF{n}$. Let $f$ be a
	Boolean function over $\GF n$ and $\alpha$ be a non-zero $1$-linear structure of $f$.Then  $F(y)=L(y)+L(\alpha)f(y)$ is $2$-to-$1$ on $\GF{n}$.
\end{proposition}
\begin{proof}
	For any $b\in \GF{n}$, let $F(y)=L(y)+L(\alpha)f(y)=b$. Then we have
	\begin{equation}
	\left\{\begin{array}{lll}
	L(y)&=&b,\\
	f(y)&=&0; 
	\end{array}\right.
	\end{equation}
or 
	\begin{equation}\label{yba}
\left\{\begin{array}{lll}
L(y)&=&b+L(\alpha),\\
f(y)&=&1.  
\end{array}\right.
\end{equation}
It follows from Eq. (\ref{yba}) and $L$ is a linear permutation that $y=\alpha+L^{-1}(b)$. Then $f(y)=f(\alpha+L^{-1}(b))=f(L^{-1}(b))+1$ 
since $\alpha$ is a non-zero $1$-linear structure of $f$. Hence 
$F(y)=b$ has either zero or two solutions in $\gf_{2^n}$, which completes the proof. 
\end{proof}

\mqu{
\subsection{APN functions}
\emph{ Almost perfect nonlinear} (APN) functions are  important research objects in cryptography and coding theory. Let us recall their definition.
\begin{definition}
	Let $F$  be a mapping from $\GF n$ to itself ($n$ a positive integer). The function $F$ is said to be APN
	if $$\max_{a\in\GF n^\ast}\max_{b\in\GF n}\#\{x\in\GF n\mid F(x+a)+F(x)=b\}=2.$$
\end{definition}

It is clear that a function $F$ over $\gf_{2^n}$ is APN if and  only if 
$D_aF(x)=F(x+a)+F(x)$ is  $2$-to-$1$ over $\gf_{2^n}$ for every $a\in \gf_{2^n}^\ast$. Hence one can construct a big family of $2$-to-$1$ mappings from an APN function. For the known list of APN functions over $\gf_{2^n}$, please refer to \cite{Pott}\cite{Villa} and the references therein. From these APN functions, we can construct plenties of $2$-to-$1$ mappings  over $\gf_{2^n}$.

Conversely, two-to-one mappings over finite fields in characteristic $2$ can also allow to construct APN functions as follows: let $G:\gf_{2^n}\rightarrow \gf_{2^n}$ be a mapping such that $F_a(x):=G(x+a)+G(x)$ is $2$-to-$1$ over $\gf_{2^n}$ for every $a\in \gf_{2^n}^\ast$. Then $G$ is an APN function.
}


\section{ $2$-to-$1$ polynomial mappings in classical classes of polynomials}\label{Pol}
\subsection{Linearized polynomials and Monomials}
{
Firstly, we have the following general proposition characterizing $2$-to-$1$ linear mappings over finite fields in even characteristic.
\begin{proposition}
	  Let $L$ be \mqu{an $\gf_2$-}linear mapping from $\gf_{2^n}$ to $\gf_{2^m}$. Then $L$ is a $2$-to-$1$ mapping if and only if $\dim \ker L=1$.
\end{proposition}

}

We have to mention that the simplest example of $2^m$-to-$1$ mapping is the trace  function $\Tr_{2^{2m}/2^m}$ from $\mathbb{F}_{2^{2m}}$ to
$\mathbb{F}_{2^m}$.

On the other hand, there are other explicit constructions of $2$-to-$1$ mapping.

\begin{proposition}{\cite[Theorem9]{CharpinKyureghyanFFA}}
Let $n$ be an odd prime number satisfying {one of } the two following conditions (where $ord_n(2)$ denotes the order of 2 modulo $n$, that is the smallest positive integer $k$ such that $n$ divides $2^k-1$)
\begin{itemize}
\item $ord_n(2)=n-1$;
\item $n=2r+1$, $r$ odd and { $ord_n(2)=r$}.
\end{itemize}
Let I be a nonempty set of integers in the range {[1, $\frac{n-1}{2}$]}. Then, for any such I the mapping $L_{I}(x)=\sum_{i\in I}(x^{2^i}+x^{2^{n-i}})$ is $2$-to-$1$ with kernel $\{0,1\}$.
\end{proposition}

\begin{proposition}
Let $F_k$ be the mapping from $\mathbb{F}_{2^n}$ to $\mathbb{F}_{2^n}$ given
by $F_k(x)=x^{2^k+1}+x^{2^k+2^m}$ where $n=2m$ and $0<k<m$. Let $G_{k, a}(x)=F_k(x+a)+F_k(x)$. If $a\in\mathbb{F}_{2^n}\setminus\mathbb{F}_{2^m}$ and $\gcd(k,m)=1$, then
$x\mapsto G_{k, a}(x)$ is $2$-to-$1$. 
\end{proposition}
\begin{proof} We have
$G_{k, a}(x)=F_k(a)+L_a(x)$, where $$L_a(x)=a^{2^k}x^{2^m} + (a+a^{2^m})x^{2^k} + a^{2^k}x$$ is a linearized polynomial. Then it suffices to prove that $L_a(x)=0$ has exactly two zeros in $\gf_{2^n}$ if $a\in\mathbb{F}_{2^n}\setminus\mathbb{F}_{2^m}$ and $\gcd(k,m)=1$. 
Clearly, if $x\in \gf_{2^m}$, then $L_a(x)=0$ reduces to $x=0$. 
Now we assume that  $x\notin \gf_{2^m}$. 
Then $$0=L_a(x)^{2^m} = a^{2^{k+m}}x^{2^m} + (a+a^{2^m})x^{2^{k+m}} + a^{2^{k+m}}x .$$ Adding the above two equations leads to 
$$ (a^{2^{k+m}} + a^{2^k}) (x^{2^m} + x) + (a+a^{2^m}) (x^{2^{k+m}} + x^{2^{k}} ) = 0. $$
Hence $(x^{2^m} + x)^{2^k-1}=(a^{2^m} + a)^{2^k-1}$, which further leads to $x^{2^m} + x = a^{2^m} + a$ since $\gcd(k, m)=1$. 
Let $x=a+y$. Then $y\in \gf_{2^m}$. Plugging it into $L_a(x)=0$, we get $y=0$, 
which means $x=a$. Thus $L_a(x)=0$ has exactly two zeros $x=0$ and $x=a$ in $\gf_{2^n}$. \end{proof}

Now we recall the following trivial characterization of $2$-to-$1$ monomial mapping over $\gf_q$. 

\begin{proposition}
	Let $f(x)=ax^d$ be a monomial polynomial over $\gf_{q}$, where $a\neq 0$. Then $f$ is $2$-to-$1$ over $\gf_{q}$ if and only if $\gcd(d, q-1)=2$. 
\end{proposition}

{Then we recall a result which is closely related to the monomial mapping. } In 1998 Maschietti discovered a class of cyclic difference sets with Singer parameters which was called the hyperoval sets \cite{Maschietti98}. Let $m$ be odd. Maschietti showed that $$ M_k:=\gf_{2^m}\setminus\{x^k+x: x \in\gf_{2^m} \}$$
is a difference set \mqu{if and only if } $x \rightarrow  x^k$ is a permutation on $\gf_{2^m}$  and the mapping $x \rightarrow  x^k+x$   is two-to-one. {The following $k$ yields difference sets, hence they also yields $2$-to-$1$ mappings.}

\begin{proposition}
	Let $m$ be odd. Then $f(x)=x^k+x$ is $2$-to-$1$ over $\gf_{2^m}$ if one of the following case holds.
	
	\begin{enumerate}
		\item $k= 2$ (the Singer case);
		\item $k= 6$ (the Segre case);
		\item $k= 2^\sigma + 2^\pi$  with $\sigma =(m+1)/2$ and $4\pi \equiv 1 \mod m$ (the Glynn I case);
		\item $k= 3\cdot 2^\sigma + 4$  with $\sigma =(m+1)/2$ (the Glynn II case).
	\end{enumerate} 
\end{proposition}

\subsection{Low degree polynomials}

Let $f(x)=\sum_{i=0}^{n}a_ix^i\in\gf_q[x]$, where $q=p^m$. In this subsection, we consider $2$-to-$1$ mappings of degree $\leq 4$ over $\gf_q$. It is clear that $f(x)\in\gf_q[x]$ is a $2$-to-$1$ mapping over $\gf_q$ if and only if so is  $f_1(x)=bf(x+c)+d$, where $b,c,d\in \gf_q$ with $b\neq0$. Hence, W.L.O.G, we consider $f(x)\in\gf_q[x]$ with \emph{normalized form}, i.e., $f(x)$ is monic ($a_n=1$), $f(0)=0$ ($a_0=0$), and when $\gcd(p,n)=1$, the coefficient of $x^{n-1}$ is $0$ ($a_{n-1}=0$).

\textbf{(A) $n\le3$.}

When $n=1$, $f(x)=x$ can not be a $2$-to-$1$ mapping. 

When $n=2$, let $f(x)=x^2+ax\in\gf_q[x]$, where $q=p^m$.  If $p=2$, then $f(x)$ is a $2$-to-$1$ mapping if and only if $a\neq0$. If $p\neq2$, then $f(x)$ is always a $2$-to-$1$ mapping. 

When $n=3$, consider $f(x)=x^3+a_2x^2+a_1x\in \gf_q[x]$, where $q\geq 5$. 
\mqu{Let $b\in\gf_q$. Since $f$ is a cubic polynomial,  $f(x)=b$  has generally either $0$ or $1$ or $3$ solutions in $\gf_q$.   And $f(x)=b$ has two solutions in $\gf_q$ if and only if one of its solutions has multiplicity $2$, while this case only occurs for at most two values of $b$ when $a_2$ and $a_1$ are fixed. Hence $f$ can not be    $2$-to-$1$ if $q\geq 7$. It is shown by an exhaustive search that there exists ten $2$-to-$1$ polynomials with such form over  $\gf_5$ : $f(x)=x^3\pm 2x$, $f(x)=x^3\pm x^2+4x$, 
	$f(x)=x^3\pm 2 x^2+x$, and $f(x)=x^3+c x^2$, $c\in \gf_5^\ast$. 	
}

\textbf{(B) $n=4$.}

We divide the discussion into three cases according to the characteristic of the field. In more details, for the cases of $p=2$, $p=3$ and $p>3$. 
The following lemma will be needed. 
\begin{lemma}
	\label{lem_deg_3}\cite{BRS}\cite{W}
	Let $a, b \in \gf_{q}$, where $q=p^m$ and $b\neq0$. Then the cubic equation $x^3+ax+b=0$ has a unique solution in $ \gf_{q}$ if and only if
	one of the following holds
	\begin{enumerate}
		\item  $p=2$ and $\Tr_{p^m/p} \left(\frac{a^3}{b^2}\right) \neq \Tr_{p^m/p}(1)$;
		\item $p=3$, $a=0$ or \mqu{$-a$} is a non-square in $\gf_q$; 
		\item $p>3$, $-4a^3-27b^2$ is a non-square in $\gf_q$.
	\end{enumerate}
\end{lemma}

\textbf{(B.1) $p=2$.}

\begin{theorem}
	Let $q=2^m$ and $f(x)=x^4+a_3x^3+a_2x^2+a_1x\in\gf_q[x]$. Then $f(x)$ is $2$-to-$1$ if and only if one of the following holds:
	\begin{enumerate}
		\item $a_3=a_1=0$, $a_2\neq0$;
		\item $a_3=0, a_1\neq0$ and $\Tr_{2^m/2}\left(\frac{a_2^3}{a_1^2}\right)\neq\Tr_{2^m/2}(1)$;
		\item $m$ is odd, $a_3\neq0$ and $a_2^2=a_1a_3$.
	\end{enumerate}
\end{theorem}
\begin{proof}
\mqu{First assume that $a_3=0$.  T}hen $f(x)=x^4+a_2x^2+a_1x$ is linearized. Hence $f(x)$ is $2$-to-$1$ if and only if $f(x)=x\left(x^3+a_2x+a_1\right)=0$ has exactly two solutions, which means that $x^3+a_2x+a_1=0$ has exactly one solution in $\gf_q^{*}$. 

If $a_1=0$, then $x^3+a_2x=0$, $x=0$ or $x^2=a_2$. Hence, $a_2\neq0$.

If $a_1\neq0$, then according to Lemma \ref{lem_deg_3}, we know that $x^3+a_2x+a_1=0$ has exactly one solution in $\gf_q^{*}$ if and only if $\Tr_{2^m/2}\left(\frac{a_2^3}{a_1^2}\right)\neq\Tr_{2^m/2}(1)$.

\mqu{Now} assume that $a_3\neq0$.  \mqu{Then  $f$ is $2$-to-$1$ if and only if } for any $b\in\gf_q$, we have 

\begin{equation}
\label{diff_b}
0=f(x+b)+f(b)=x^4+a_3\left(x^3+bx^2+b^2x\right)+a_2x^2+a_1x
\end{equation}
has exactly two solutions $x=0$ and $x=x_0\in\gf_q^{*}$,  \mqu{or equivalently, } 
\begin{equation}
\label{x_b_p21}
x^3+a_3x^2+\left(ba_3+a_2\right)x+\left(b^2a_3+a_1\right)=0.
\end{equation}
has exactly one solution in $\gf_q^{*}$ for \mqu{any} $b\in\gf_q$.

If $b^2a_3+a_1=0$, i.e., $b=a_1^{1/2}a_3^{-1/2}$, then Eq. (\ref{x_b_p21}) reduces to 

\begin{equation}
\label{x_b_p22}
x^2+a_3x+\left(ba_3+a_2\right)=0.
\end{equation}

Since Eq. (\ref{x_b_p22}) has exactly one solution in $\gf_q^{*}$, we have $ba_3+a_2=0$. Then with  $b^2a_3+a_1=0$, we know  that $a_2^2=a_1a_3.$  
Plugging it into Eq. (\ref{x_b_p21}), and letting $x=a_3+z$, we get
\begin{equation}
\label{z_b_1}
z^3+uz+v=0,
\end{equation}
where $u=a_3^2+ba_3+a_2$ and $v=b^2a_3+ba_3^2+a_2a_3+a_1.$ 

If $v=0$, i.e., $\left(ba_3^{1/2}+a_1^{1/2}\right)\left(ba_3^{1/2}+a_1^{1/2}+a_3^{3/2}\right)=0$, then $b=a_1^{1/2}a_3^{-1/2}$ or $b=a_1^{1/2}a_3^{-1/2}+a_3$.  When $b=a_1^{1/2}a_3^{-1/2}$, Eq. (\ref{x_b_p21}) has exactly one \mqu{nonzero} solution $x=a_3$. When $b=a_1^{1/2}a_3^{-1/2}+a_3$,  we have $u=0$. Therefore, $z=0$ is the unique solution of Eq. (\ref{z_b_1}). Moreover, Eq. (\ref{x_b_p21}) also has  exactly one solution $x=a_3$. 

Now we assume that  $v\neq0$.  According to Lemma \ref{lem_deg_3}, it suffices to prove   

$$\Tr_{2^m/2}\left(u^3/v^2\right)\neq\Tr_{2^m/2}(1).$$

{\bfseries Claim 1:} $\Tr_{2^m/2}\left(u^3/v^2\right)=0$, where  $u=a_3^2+ba_3+a_2$, $v=b^2a_3+ba_3^2+a_2a_3+a_1\neq 0$ and $a_2^2=a_1a_3$. 

\mqu{We have}

\begin{eqnarray*}
	\Tr_{2^m/2}\left(\frac{u^3}{v^2}\right) &=& \Tr_{2^m/2}\left(\frac{\left(a_3^2+ba_3+a_2\right)^3}{b^4a_3^2+b^2a_3^4+a_2^2a_3^2+a_1^2}\right) \\
	&=& \Tr_{2^m/2}\left(\frac{\left(a_3^4+b^2a_3^2+a_1a_3\right)^3}{b^4a_3^4\left(b^4+a_3^4\right)+a_1^2\left(a_3^3+a_1\right)^2}\right) \\
	&=& \Tr_{2^m/2}\left(\frac{A}{B^2}\right),
\end{eqnarray*}
where $$B=b^2a_3^2\left(b^2+a_3^2\right)+a_1\left(a_3^3+a_1\right)$$ 
and $$A=\left(b^2a_3^4+a_3^6+a_1a_3^3\right)B+\left(b^2a_3^4+a_3^6+a_1a_3^3\right)^2.$$
Let $U=b^2a_3^4+a_3^6+a_1a_3^3$. Then $A=UB+U^2.$ Thus,  

$$\Tr_{2^m/2}\left(\frac{u^3}{v^2}\right)=\Tr_{2^m/2}\left(\frac{A}{B^2}\right)=\Tr_{2^m/2}\left(\frac{UB+U^2}{B^2}\right)=\Tr_{2^m/2}\left(\frac{U}{B}\right)+\Tr_{2^m/2}\left(\frac{U^2}{B^2}\right)=0.$$

From Claim 1, we know that $\Tr_{2^m/2}\left(u^3/v^2\right)\neq\Tr_{2^m/2}(1)$ if and only if $m$ is odd. The proof is completed.
\end{proof}

\textbf{(B.2) $p=3$.}

\begin{theorem}\label{th_d4p3}
	Let $f(x)=x^4+a_2x^2+a_1x\in\gf_{3^m}[x]$, where $m>1$. Then $f(x)$ is $2$-to-$1$ over $\gf_{3^m}$ if and only if $a_2=a_1=0$ and $m$ is odd.
\end{theorem}

Before proving this theorem, we give two lemmas. 

\begin{lemma}(\cite[Theorem 6.2.2]{MP13})
	\label{char_le}
	Let $f\in\gf_q[x]$ be a polynomial of degree $d>0$ and $\chi: \gf_q^{*}\to\mathbb{C}^{*} $ a non-trivial multiplicative character of order $m$ (extended by zero to $\gf_{q}$). If $f$ is not an $m$-th power in $\overline{\gf}_q[x]$,  then
	$$\left|\sum_{x\in\gf_{q}}\chi(f(x))\right|\le (d-1)\sqrt{q}.$$
\end{lemma}

\begin{lemma}
	\label{h(x)_square}
	Let $h(x)=x^6+2a_2x^4+a_1x^3+a_2^3\in\gf_{q}[x]$, where $q=3^m$. Then $h(x)$ is a square in $\overline{\gf}_q[x]$ if and only if $a_1=a_2=0.$
\end{lemma}

\begin{proof}
	Let $h(x)=g(x)^2$, where $g(x)\in\overline{\gf}_q[x]$. Hence, $h'(x)=2g(x)g'(x)$, it follows that $g(x)\mid \gcd\left(h(x),h'(x)\right)$. On the other hand, $h'(x)=2a_2x^3$. Therefore, if $a_2\neq 0$, then  $\gcd\left(h(x),h'(x)\right)=\gcd\left(x^6+2a_2x^4+a_1x^3+a_2^3,2a_2x^3\right)= 1$, which is impossible. Thus $a_2=0$ and $h(x)=x^3(x^3+a_1)$.  Then it follows that   $a_1=0$, $g(x)=x^3$ and $h(x)=x^6$. We are done.  
\end{proof}

\mqu{ \bfseries Proof of Theorem \ref{th_d4p3}.}
 Assume that $f$ is $2$-to-$1$. If $a_2=a_1=0$, then $f(x)=x^4$ is $2$-to-$1$ over $\gf_q$ if and only if $\gcd(4, 3^m-1)=2$, or equivalently, if and only if $m$ is odd. 

Now we assume that $(a_1, a_2)\ne (0, 0)$. 

\mqu{Since  $f$ is $2$-to-$1$,  then for all but one $b$ in $\gf_q$,}
$$0=f(x+b)-f(b)=x^4+bx^3+b^3x+a_2\left(2bx+x^2\right)+a_1x$$
has exactly two solutions $x=0$ or $x=x_0\in\gf_q^{*}$, which means 
\begin{equation}
\label{x_b_1}
x^3+bx^2+a_2x+\left(b^3+2a_2b+a_1\right)=0
\end{equation}
has exactly a unique solution in $\gf_q^{*}$  \mqu{for all but one $b$ in $\gf_q$.
For convenience, we denote by $b_0$ this exceptional element. }


\mqu{Now let $b\neq 0$ and $b\neq b_0$.}  Let $g(x)=x^3+bx^2+a_2x+\left(b^3+2a_2b+a_1\right)$, and let  
$$\tilde{g}(x)=x^3g\left(\frac{1}{x}+\frac{a_2}{b}\right)=A_3x^3+bx+1,$$
where $A_3=\frac{a_2^3}{b^3}+b^3+2a_2b+a_1.$ It is clear that $g(x)$ has exactly one solution in $\gf_q^{*}$ if and only if $\tilde{g}(x)$ has exactly one solution.
Since $\tilde{g}(x)$ is affine over $\gf_{3^m}$,  it has exactly one solution  in $\gf_q^\ast$  if and only if $A_3=0$ or $\eta\left(-\frac{b}{A_3}\right)=-1$ if $A_3\neq0$. Therefore, when $A_3\neq0$,  $\eta\left(-\frac{b}{A_3}\right)=-1.$ Furthermore, 
$$\eta\left(-\frac{b}{A_3}\right)=\eta(-1)\eta\left(b^6+2a_2b^4+a_1b^3+a_2^3\right)=-1.$$

Let $h(b)=b^6+2a_2b^4+a_1b^3+a_2^3$. Then the above discussion leads to 
$$\left|\sum_{b\in\gf_{q}}\eta(h(b))\right|\ge q-6-2\cdot 2=q-10.$$
On the other hand, according to Lemmas \ref{char_le} and \ref{h(x)_square}, we have $h(b)$ is not a square and 
$$\left|\sum_{b\in\gf_{q}}\eta(h(b))\right|\le 5\sqrt{q}.$$

Thus, $q-10\le 5\sqrt{q}$, where $q=3^m$. Hence we have $m\le3$. An \mqu{exhaustive} search over $\gf_{3^2}$ and $\gf_{3^3}$ found that there is no $2$-to-$1$ function with the form of $f(x)=x^4+a_2x^2+a_1x$, where $\left(a_2,a_1\right)\neq(0,0)$. The proof is completed.$\hfill\Box$

\textbf{(B.3)  $p\geq 5$. }

\begin{theorem}
	Let $q=p^m$, where $p\geq 5$ and $f(x)=x^4+a_2x^2+a_1x\in\gf_q[x]$. Then $f(x)$ is $2$-to-$1$ if and only if one of the following holds:
		\begin{enumerate}
			\item $a_1=a_2=0$, $\gcd(4,q-1)=2$, i.e., $q\equiv3\pmod4$;
			\item $q=5$, $f(x)=x^4+x^2\pm2x$, or $f(x)=x^4-x^2\pm x$ or $f(x)=x^4\pm2x^2$;
			\item $q=7$, $f(x)=x^4\pm2x$.
		\end{enumerate}
\end{theorem}
\begin{proof}
Assume that $f$ is $2$-to-$1$. If $a_2=a_1=0$, then $f(x)=x^4$ is $2$-to-$1$ over $\gf_q$ if and only if $\gcd(4, q-1)=2$, or equivalently, if and only if  $q\equiv3\pmod4$. 

Now we assume that $(a_1, a_2)\ne (0, 0)$. 
Then for \mqu{all but one $b\in\gf_q$,} $f(x+b)-f(b)=x^4+4bx^3+6b^2x^2+4b^3x+a_2\left(2bx+x^2\right)+a_1x=0$ has exactly two solutions in $\gf_q$, i.e., 
\begin{equation}
\label{subcase_2.3_x_b}
x^3+4bx^2+\left(6b^2+a_2\right)x+\left(4b^3+2a_2b+a_1\right)=0
\end{equation}
has exactly one solution in $\gf_{q}^{*}$ \mqu{for all but one $b\in\gf_q$.} 

Let $x=y-\frac{4}{3}b$. Plugging it into Eq. (\ref{subcase_2.3_x_b}), we get 
\begin{equation}
\label{subcase_2.3_y_b}
y^3+A_1y+A_0=0,
\end{equation}
where $$A_1=\frac{2}{3}b^2+a_2,$$
and $$A_0=\frac{20}{27}b^3+\frac{2}{3}a_2b+a_1.$$

Since Eq. (\ref{subcase_2.3_y_b}) has exactly one solution in $\gf_q^{*}$ \mqu{for all but one $b\in\gf_q$}, $\Delta=-4A_1^3-27A_0^2$ is a nonsquare. In addition, after computing, we obtain 
$$\Delta=-4A_1^3-27A_0^2=-16h(b),$$
where $h(b)=b^6+2a_2b^4+\frac{5}{2}a_1b^3+\frac{5}{4}a_2^2b^2+\frac{9}{4}a_1a_2b+\frac{1}{4}a_2^3+\frac{27}{16}a_1^2.$ Assume $h(b)=g(b)^2$, where $g(b)=b^3+g_2b^2+g_1b+g_0\in \bar{\gf}_q[b]$. Then $g_2=0$, $2a_2=2g_1$, i.e., $g_1=a_2$. Moreover, $$h(b)=\left(b^3+a_2b+g_0\right)^2=b^6+2a_2b^4+2g_0b^3+a_2^2b^2+2a_2g_0b+g_0^2.$$
After matching the coefficients of the above equation, we know that $a_2=0$ and $a_1=0$. Hence, $h(b)$ is not a square since $(a_1, a_2)\neq (0, 0).$

Similarly, on one hand, since $\Delta=-16h(b)$ is a non-square for \mqu{all but at most one} $b\in \gf_q^\ast$ such that $A_0\neq 0$, we have 

	$$\left|\sum_{b\in\gf_{q}}\eta(h(b))\right|\ge q-2\cdot5=q-10.$$
On the other hand,  since $h(b)$ is not a square, it follows from Lemma
\ref{char_le}  that 
$$\left|\sum_{b\in\gf_{q}}\eta(h(b))\right|\le 5\sqrt{q}.$$
Thus $q=p^m<49$. 
An \mqu{exhaustive} search finishs the proof.
\end{proof}

\subsection{Dickson polynomials}
The Dickson polynomials have been extensively investigated in recent years 
under different contexts.
\begin{definition}
The Dickson polynomial of the first kind of degree $n$ in indeterminate $x$ and
with parameter $a\in\mathbb{F}_q^*$ is defined by Waring's formula 
\begin{equation}\label{eq-D}
D_n(x,a)=\sum_{i=0}^{ \lfloor n/2\rfloor }  \frac{n}{n-i} 
 \left(\begin{array}{c}n-i
 \\
 i\end{array}
\right) a^n x^{n-2i},~n\geq 1.
\end{equation}
\end{definition}
\begin{proposition}\cite{Dickson}
	Let $D_n(x, a)$ be the  Dickson polynomial of the first kind. Then 
	 $D_n(x, a)$ is $e$-to-$1$ over $\gf_q$ if and only if $\gcd(n, q^2-1)=e$. 
\end{proposition}

\begin{remark}
The previous proposition can only provide $2$-to-$1$ mappings if $q$ is odd.  
\end{remark}

\mqu{In 2009 Hou et al. considered a different perspective of the Dickson polynomial \cite{Hou-1}. They fixed $a\in\gf_q$, and studied  the polynomial $D_n(a, x) \in\gf_ q [x]$, which they called reversed Dickson polynomial. }
The following result characterizes the \mqu{reversed} Dickson polynomial which are permutation in even characteristic in terms of
 $2$-to-$1$ mappings.
\begin{proposition}\cite[Proposition 4.2]{Hou-1}
\mqu{$D_n(1, x)$} is  permutation polynomials over $\mathbb{F}_{2^m}$ if and only if the function $y\mapsto y^n-(1-y)^n$ is a $2$-to-$1$ mapping on $\mathbb{F}_{2^m} \cup V$ where $V:=\{x\in \mathbb{F}_{2^{2m}}\mid x^{2^m}=1-x\}$.
\end{proposition}

\subsection{Muller-Cohen-Matthews polynomials}
\begin{definition}
	Let $q=2^m$ where $m>1$ is a positive integer. Let $T_k(x):=\sum_{i=0}^{k-1} x^{2^i}$. Then $f_{k,d}(x):=
	\frac{T_k^d(x^c)}{x^{2^k}}$ (where $cd=2^k+1$) is the so-called \emph{Muller-Cohen-Matthews polynomial} in $\mathbb{F}_{2^n}[x]$.
\end{definition}
For every odd $k$, all $f_{k,d}$ are exceptional polynomials which induce a permutation
on $\mathbb{F}_{2^{m}}$ when $m$ is relatively prime to $k$.  We shall apply Muller-Cohen-Matthews polynomials
for the choice $c=1$ and $d=2^k+1$.  
Then we have:
\begin{proposition}\cite{Dillon-Dobbertin}
Suppose that $\gcd(k, m)=1$ and  $k$ is even.  Then $f_{k, 2^k+1}$ is a $2$-to-$1$ mapping on $\mathbb{F}_{2^{m}}$.
\end{proposition}

{
\subsection{Other Polynomials}
We recall results on construction of $2$-to-$1$ mappings {related with the trace functions}.
\begin{proposition}\cite{CharpinKyureghyanFQ}
	Let $0\leq i\leq n-1$, $i\not\in \{0, \frac{n}2\}$ and $\delta$, $\gamma\in\GF n$ be such that $\delta^{2^i-1}=\gamma^{1-2^{2i}}$.
	If $\Tr_{2^n/2} (\delta\gamma^{2^i+1})=1$,  then $F: \mathbb{F}_{2^n}\rightarrow \mathbb{F}_{2^n}$ defined by
	$F(y)=y+\gamma  \Tr_{2^n/2}  (\delta y^{2^i+1})$ is $2$-to-$1$.
\end{proposition}

\begin{proposition}\cite{CharpinKyureghyanFQ}
	Let $\gamma\in\mathbb{F}_{2^n}^\ast$ such that $\Tr_{2^n/2} (\gamma)=1$. Then the mapping $x\mapsto x^s+\gamma \Tr_{2^n/2}(x^t)$ is $2$-to-$1$ {over $\GF n$,} where $s$ and $t$ are two positive integers.
\end{proposition}
}

\section{Applications of $2$-to-$1$ mappings over finite fields }\label{appli}

\subsection{Bent functions}

Bent functions introduced in 1974 (\cite{Dillon74},\cite{Rothaus76}) are extremal objects in combinatorics and Boolean function theory. They are maximally nonlinear Boolean functions.
Recall that the \emph{nonlinearity} of a Boolean function $f$, denoted by
$nl(f)$, is defined as  the minimum Hamming distance between $f$ and all
affine functions (that is, of degree at most $1$). It can be
expressed by means of the Walsh transform as follows: 
$$ nl(f)=2^{n-1}-\frac{1}{2} \max_{b\in \GF n} {|\widehat{\chi}_f(b)|}.$$
Because of the well-known Parseval's relation $\sum_{b\in
\GF n}{\widehat{\chi}_f (b)}^2 = 2^{2n}$, $ nl(f)$ is upper bounded by
$2^{n-1}-2^{n/2-1}$. This bound is tight for $n$ even.

\begin{definition}
Let $n$ be an even integer. A Boolean function on $\GF n$ is
said to be bent if the upper bound $2^{n-1}-2^{n/2-1}$ on its
nonlinearity $nl (f)$ is achieved with equality.
\end{definition}
Bent functions on $\GF n$ exist then only when $n$ is even.
We have the following main characterization of the
bentness for Boolean functions in terms of the Walsh transform.

\begin{proposition}
Let $n$ be an even integer. A Boolean function $f$ is  bent if and only
if its Walsh transform satisfies $\displaystyle \widehat{\chi}_f(a) = \pm
2^{\frac{n}{2}}$ for all $a \in \GF n$.
\end{proposition}

 A  recent survey on bent functions can be found in \cite{CarletMesnagerDCC2016}. A book devoted especially to bent functions and containing a complete survey on bent
functions (including its variations and generalizations) is \cite{MesnagerBook}.

One of the important classes of bent functions is the so-called  class $\mathcal {H}$ whose elements $g$ are defined in bivariate representation over  $\GF m\times\GF m$ by 

\begin{equation}\label{e1}
  g(x,y)=\left\{\begin{array}{l}\Tr_{2^m/2}\left(x\psi\left(\frac{y}{x}\right)\right), \mbox{ if }x\neq 0;\\ \Tr_{2^m/2}(\mu y),\mbox{ if }x=0,\end{array}\right.
 \end{equation}
where $\psi: \GF m \rightarrow \GF m$ and $\mu\in\GF m$.

Two-to-one mappings over finite fields in characteristic $2$ allow to construct bent Boolean functions in bivariate representation from the class $\mathcal H$  as follows:
\begin{proposition}(\cite{CarletMesnager2011})
Let $g$ be a function defined on $\GF m\times\GF m$ by (\ref{e1}).
 Then $g$ is bent if and only if
 \begin{equation}\label{cond1}
G(z):=\psi(z)+\mu z \mbox{ is a  permutation on }  \GF m
 \end{equation}
 \begin{equation}\label{cond2}
 \forall \beta\in \GF m^\ast ,\mbox{the function } z\mapsto G(z)+\beta z \mbox{ is 2-to-1 on }  \GF m.
 \end{equation}
\end{proposition}
The following result shows that one can construct vectorial bent functions from certain two-to-one mappings.
\begin{theorem}
Let $m$ and $k$ be two positive integers such that $\gcd(k,2^m-1)=1$.
Assume that $z\mapsto z^{k}+bz$ is 2-to-1  on $\GF m$, where $b\in\mathbb{F}_{2^m}^\ast$. 
Then the vectorial function $F$ defined from $\mathbb{F}_{2^m}\times \mathbb{F}_{2^m}$ to $\mathbb{F}_{2^m}$ by
$F(x,y)=y^k x^{k(2^m-2)+1}$ is bent.

\end{theorem}
\begin{proof}
Recall that $F$ is bent if and only if all its components (Boolean) functions $F_v: x\mapsto \Tr_{2^m/2}(vF(x,y))$ ($v\in\mathbb{F}_{2^m}^\ast$)  are bent on $\GF m$. Let us compute the Walsh transform of $F_v$ at each element $(a, b)\in \mathbb{F}_{2^m}\times \mathbb{F}_{2^m}$. 

 We have:
\begin{displaymath}
\begin{split}
  \widehat{\chi}_{F_v}(a,b)&=\!\!\!\sum_{x\in\GF m}\!\!\!\sum_{y\in\GF m} (-1)^{\Tr_{2^m/2} (vy^k x^{k(2^m-2)+1})+\Tr_{2^m/2} (ax)+\Tr_{2^m/2} (by)}\\
  &=\sum_{y\in\GF m}\!\!\!(-1)^{\Tr_{2^m/2} (by)}+\!\!\!\sum_{x\in\GF m^\ast}\!\!\!\sum_{y\in\GF m}  (-1)^{\Tr_{2^m/2} (vy^k x^{k(2^m-2)+1})+\Tr_{2^m/2} (ax)+\Tr_{2^m/2}(by)}\\
  &=2^m\delta_0(b)+\!\!\!\sum_{x\in\GF m^\ast}\!\!\!\sum_{y\in\GF m}  (-1)^{\Tr_{2^m/2} (vy^k x^{k(2^m-2)+1})+\Tr_{2^m/2} (ax)+\Tr_{2^m/2} (by)},
\end{split}
\end{displaymath}
where $\delta_0(b)$ denotes $1$ if $b=0$ and $0$ if $b\neq0$.

Set $z:=x^{(2^m-2)}y$, that is, $y=zx$. Therefore,
\begin{displaymath}
\begin{split}
  \widehat{\chi}_{F_v}(a,b)&=2^m\delta_0(b)+\!\!\!\sum_{z\in\GF m}\!\!\!\sum_{x\in\GF m^\ast}  (-1)^{\Tr_{2^m/2} (vz^kx)+\Tr_{2^m/2} (ax)+\Tr_{2^m/2} (bzx)}\\
  &=2^m\delta_0(b)+\!\!\!\sum_{z\in\GF m}\!\!\!\sum_{x\in\GF m^\ast}  (-1)^{\Tr_{2^m/2} ((vz^k+bz+a)x)}\\
   &=2^m\delta_0(b)-2^m+2^m \#\{ z\in\GF m \mid vz^k+bz+a=0\}.
     \end{split}
\end{displaymath}

Now, if $b\not=0$ then the equation $vz^k+bz+a=0$ has $0$ or $2$ solutions in $\GF m$ since the mapping $z\mapsto z^{k}+bz$ is 2-to-1  on $\GF m$. Thus, $\widehat{\chi}_{F_v}(a,b)=2^m(\#\{ z\in\GF m \mid vz^k+bz+a=0\}-1)=\pm 2^m$.\\
If $b=0$ then the equation $vz^k+a=0$ has only one solution in $\GF m$ since $\gcd(k, 2^m-1)=1$. Hence, $\widehat{\chi}_{F_v}(a,b)=2^m\#\{ z\in\GF m \mid vz^k+a=0\}=2^m$. This completes the proof.
\end{proof}

\subsection{Semi-bent functions}
Semi-bent functions (or $2$-plateaued functions) on $\GF n$ exist only when $n$ is even. Semi-bent functions are defined as follows.

\begin{definition}
Let $n$ be an even integer. A Boolean function on $\GF n$ is
said to be semi-bent  if its Walsh transform satisfies $\displaystyle \widehat{\chi}_f(a) \in\{0, \pm  2^{\frac{n+2}2}\}$
 for all $a \in \GF n$.
\end{definition}
Recall that the Maiorana-McFarland's constructions are the best known primary constructions of bent functions   (\cite{Maiorana73, Dillon74}).
The \emph{Maiorana-McFarland class} is the set of all the Boolean functions on $\GF m\times \GF m$ 
of the form : $f(x,y) = \Tr_{2^m/2}(x  \pi (y)) +  g(y),$ where $ x, y \in \GF m, \pi$ is any permutation on~$\GF m$ and $g$ is any 
Boolean function on~$\GF m$. Any such function is bent (the bijectivity of~$\pi$ is a 
necessary and sufficient condition for $f$ being bent). By computing the Walsh transform, we see that $f$ is semi-bent on $\GF m\times \GF m$ \mqu{if} $\pi$ is a 2-to-1 mapping
from $\GF m$  to $\GF m$.
Therefore two-to-one mappings over finite fields in characteristic $2$ allow to construct semi-bent Boolean functions in bivariate representation from the Maiorana-McFarland class  as follows:

\begin{theorem}\label{semibent}
Let $\pi$ be a mapping from $\GF m$ to $\GF m$ and 
 $g$ be a  Boolean function on~$\GF m$.
Let $f$ be a Boolean function defined over $\GF m\times \GF m$ by $f(x,y)=\Tr_{2^m/2}(x \pi (y))+g(y)$. 
\mqu{If $\pi$ is $2$-to-$1$ on $\GF m$, then $f$ is semi-bent.}
\end{theorem}
\begin{proof}
For every $(a,b)\in \GF m\times\GF m$,
  we have:
\begin{displaymath}
\begin{split}
  \widehat{\chi}_{f}(a,b)&=\!\!\!\sum_{x\in\GF m}\!\!\!\sum_{y\in\GF m} (-1)^{\Tr_{2^m/2} (x \pi (y))+g(y)+\Tr_{2^m/2} (ax)+\Tr_{2^m/2} (by)}\\
  &=\sum_{y\in\GF m}(-1)^{g(y)+\Tr_{2^m/2} (by)}\sum_{x\in\GF m} (-1)^{\Tr_{2^m/2} ((\pi (y)+a)x)}\\
  &=2^m\sum_{y\in\GF m\mid \pi (y)=a }(-1)^{g(y)+\Tr_{2^m/2} (by)}.\\
\end{split}
\end{displaymath}
Since the mapping $y\in\GF
m \mapsto \pi (y)$ is 2-to-1 for every $a\in\GF m$, we have
$\widehat{\chi}_{f}(a,b)\in\{0, \pm 2^{m+1}\}$, which completes the proof.
\end{proof}

The following statement illustrates an example of constructions of semi-bent functions via $2$-to-$1$ mappings in the line of the Maiorana-McFarland's method.
\begin{proposition}
Let $r$ be a positive integer. Set $m=2r-1$.  Let $g$ be any Boolean function over $\GF m$. Define over $\GF m \times \GF m$ a Boolean function by  $f(x,y)=\Tr_{2^m/2} (xy^{2^r+2}+ xy )+g(y)$, $\forall (x,y)\in \GF m\times \GF m$.
Then $f$ is semi-bent.
\end{proposition}
\begin{proof}
The construction comes from Theorem \ref{semibent} and the fact that the mapping $y\in\GF
m \mapsto y^{2^r+2}+y$ is 2-to-1. (\cite{CusickDobbertin96}). 
 \end{proof}

Note that given an APN function, one can derive a  construction of semi-bent function in the sprit of  Maiorana-McFarland's method.

\subsection{Planar functions}
Let $q=p^n$ where $p$ is prime and $n$ is a positive integer.  A planar function is a function $f:\mathbb{F}_q\rightarrow \mathbb{F}_q $ such that, for every $a\in\mathbb{F}_q^*$, the function $c\mapsto f(c+a)-f(c)$ is a bijection on $\mathbb{F}_q$. Planar functions can be used to construct finite projective planes, and they have been studied by finite geometers since 1968.
 
 The following result highlights the importance of $2$-to-$1$ mappings for the constructions of planar functions of minimal size of their image set.
. 
  \begin{theorem}\cite{Kyurenghyan-Pott2008}
  Let $F :\mathbb{F}_q\rightarrow \mathbb{F}_q$ be a mapping and $\im (F)$ be its image set. Assume that $F$ is planar (which implies $\# \im (F)\geq \frac{q+1}2$). Then $F$ is $2$-to-$1$ if and only if $\# \im (F)= \frac{q+1}2$.
 \end{theorem}
 
A class of polynomials was described by Dembowski and Ostrom in \cite{DembowskiOstrom68}: the so-called Dembowski-Ostrom polynomials.
 For those polynomials, the property of being planar is equivalent to the property of being $2$-to-$1$. We first recall their definition.
 \begin{definition}\label{Dembowski-Ostrom_polynomials}
 The polynomial $P\in\mathbb{F}_q[x]$ is called a Dembowski-Ostrom polynomial if $P$ has the shape $P(x)=\sum_{i,j=0}^{n-1} a_{ij}x^{p^i+p^j}$.
 \end{definition}
 \begin{proposition}\cite{Chen-Polhill2011}
 Let $P: \mathbb{F}_q\rightarrow \mathbb{F}_q$ be given by a Dembowski-Ostrom polynomial. Then $P$ is planar if and only if $P$ is $2$-to-$1$.
 \end{proposition}
 
\subsection{Permutation polynomials}

Permutation polynomials can also be constructed from $2$-to-$1$ mappings. 
\begin{proposition}
	Let $F:\gf_{2^n}\rightarrow \gf_{2^n}$ be a two-to-one mapping. 
	Denote by $\im (F)$ the image set of $F$. Let $\phi:\im(F)
\rightarrow \gf_{2^n}\setminus \im(F)$ be a bijection, and $\gf_{2^n}=S_1\cup S_2$ be a disjoint decomposition of $\gf_{2^n}$ such that $F(S_1)=F(S_2)=\im (F)$.
Define 
	\begin{equation*}
G(x) = \left\{\begin{array}{ll}
F(x), & \text{if}\  x\in S_1;\\
\phi(F(x)), & \text{if}\  x\in S_2. 
\end{array}\right. 
\end{equation*}
	Then $G$ is a permutation polynomial over $\gf_{2^n}$. 
\end{proposition}
%
%

\section{Concluding remarks}
Many results presented in the literature highlight the importance of two-to-one mappings for designing cryptographic functions. Despite their importance, they have never been studied in detail in a general framework. Because of the gap between the interest of the notion of two-to-one mappings and the knowledge we have on it, our motivation was to bring  a systematic study on those mappings by providing several results including new tools, constructions and applications. From our criteria, we expected new constructions of cryptographic functions from two-to-one mappings.

\mqu{At last, we would like to note that most of the results of this paper can be easily generalized to $q$-to-$1$ mappings. $q$-to-$1$ mappings may also be useful in design theory, error-correcting codes, cryptography and others. We leave this generalization and the adventure to $q$-to-$1$ polynomials to interested readers.  }

 {\bf Acknowledgement.}  The authors deeply thank the Assoc. Edit. Prof. Xiaohu Tang  and the anonymous reviewers for their valuable comments and suggestions which have highly improved the manuscript.


\begin{thebibliography}{90}
\bibitem{AGW2011}
A. Akbary, D. Ghioca and Q. Wang.:
\newblock {On constructing permutations of finite fields}.
\newblock Finite Fields and Their Applications 17, pages 51--67, 2011.
	
	
		\bibitem{BRS} E.R.~Berlekamp, H.~Rumsey, and G.~Solomon.
	\newblock  On the solution of algebraic equations over finite fields. \newblock { Information And Control.}  10(67): 553-564, 1967.
	
\bibitem{Carlet2010}
C.~Carlet.:
\newblock {Boolean Functions for Cryptography and Error Correcting Codes}.
\newblock In { Chapter of the monography ``Boolean Models and Methods in
  Mathematics, Computer Science, and Engineering" published by Cambridge
  University Press, Yves Crama and Peter L. Hammer (eds.)}, pages 257--397,
  2010.	


\bibitem{Carlet2017}
C.~Carlet.: 
\newblock {Characterizations of the differential uniformity of
vectorial functions by the Walsh transform}. 
\newblock   IEEE Transactions on Information Theory 64(9), pages 6443--6453, 2018.


    \bibitem{CarletMesnager2011}
C. Carlet and S. Mesnager.:
\newblock On Dillon's class H of bent functions, Niho bent functions and o-polynomials.
\newblock J. Comb. Theory, Ser. A 118(8), pages 2392--2410, 2011.



    \bibitem{CarletMesnagerDCC2016}
C. Carlet and S. Mesnager.:
\newblock Four decades of research on bent functions.
\newblock Journal Designs, Codes and Cryptography, 78(1), pages 5--50, 2016.



    \bibitem{CharpinKyureghyanFFA}
P. Charpin and G. Kyureghyan.:
\newblock When does $G(x)+\gamma Tr(H(x))$ permute $\gf_{p^n}$?
\newblock  Finite Fields and Their Applications 15(5), pages 615--632, 2009. 

\bibitem{CharpinKyureghyanFQ}
P. Charpin and G. M. Kyureghyan.:
\newblock Monomial functions with linear structure and permutation polynomials
\newblock In Finite Fields: Theory and Applications - Fq9 - Contemporary Mathematics, AMS, number 518, pages. 99-111, 2010.

\bibitem{CharpinMesnagerSarkar}
P. Charpin, S. Mesnager and S. Sarkar.:
\newblock Involutions over the Galois field $\mathbb{F}_{2^n}$. 
\newblock IEEE Trans. Information Theory 62 (4), pages 2266-2276, 2016.

    \bibitem{Chen-Polhill2011}
Y.  Chen and J. Polhill.:
\newblock Paley type group schemes and planar Dembowski-Ostrom polynomials
\newblock Journal Discrete Mathematics, Volume 311, Issue 14, pages 1349-1364, 2011.



\bibitem{Cohen-Matthews}
S. D. Cohen and R. W. Matthews.: 
\newblock A class of exceptional polynomials, Transactions of the American Mathematical Society, Vol. 345, No. 2, pages 897-909, 1994.

\bibitem{CusickDobbertin96}
T.~W. Cusick and H.~Dobbertin.:
\newblock {Some new three-valued crosscorrelation functions for binary
  m-sequences}.
\newblock { IEEE Transactions on Information Theory 42(4)}, pages
  1238--1240, 1996.

\bibitem{DembowskiOstrom68}
P. Dembowski and T.G. Ostrom.:
\newblock {Planes of order $n$ with collineation groups of order $n^2$}.
\newblock In Math. Z. 103, pages 239--258, 1968.

\bibitem{Dillon74}
J.~Dillon.:
\newblock {Elementary Hadamard difference sets}.
\newblock In {\em PhD dissertation, University of Maryland}, 1974.

\bibitem{Dickson}
 R. Lidl, G. L. Mullen and  G. Turnwald.:
 \newblock {Dickson polynomials}
\newblock In {\em Longman Scientific Technical}, 1993.

\bibitem{Dillon-Dobbertin}
J. Dillon and H. Dobbertin.: 
\newblock New cyclic difference sets with singer parameters. Finite fields and their applications 10, pages 342- 389, 2004.

  \bibitem{Hou} 
  X. Hou.: 
  \newblock  Permutation polynomials over finite fields--A survey of recent advances, \newblock Finite Fields Appl. 32, pages 82-119, 2015.
  
    \bibitem{Hou-1} 
    X. Hou, G. L. Mullen, J. A. Sellers and J. L.Yucas.:
    \newblock
Reversed Dickson polynomials over finite fields.
\newblock Finite Fields Appl. Vol 15, Issue 6, pages 748-773, 2009.
  
  
  \bibitem{Mesnager-Oz}
  N. Kocak, S. Mesnager and F. Ozbudak.:
 \newblock Bent and semi-bent functions via linear translators. 
 \newblock Proceedings of the fifteenth International Conference on Cryptography and Coding, Oxford, United Kingdom, IMACC 2015, pages 205-224, LNCS, Springer, Heidelberg, 2015. 
  
  
\bibitem{Kyurenghyan}
G. Kyureghyan.:
\newblock Special mappings of finite fields
\newblock  Finite Fields and Their Applications. Character Sums and Polynomials. Series on Computational and applied mathematics, pages 117-240, 2013.


\bibitem{Kyurenghyan2011}
G. Kyureghyan.:
\newblock Constructing permutations of finite fields via linear translators.
\newblock Journal of Combinatorial Theory Series A. 118(3), pages 1052--1061, 2011.

 \bibitem{Kyurenghyan-Pott2008}
 G. Kyureghyan and A. Pott.:
\newblock Some theorems on planar mappings. 
\newblock Arithmetic of finite fields, 117-122, Lecture Notes in Comput. Sci., 5130, Springer, Berlin, 2008.

\bibitem{Maiorana73}
R.~L. McFarland.:
\newblock A family of noncyclic difference sets.
\newblock { Journal of Combinatorial Theory, Series A, No. 15}, pages 1--10,
  1973.
  
  \bibitem{MesnagerBook}
S. Mesnager.:
\newblock Bent functions: fundamentals and results.
\newblock Springer, Switzerland, 2016. 

\bibitem{Maschietti98}
A. Maschietti.: 
\newblock {Difference sets and hyperovals}.
\newblock {}Des. Codes Cryptgr. 14 (1998) 89–98.

\bibitem{MP13}
 G.L. Mullen, D. Panario.:
 \newblock{Handbook of Finite Fields},  Taylor  Francis, Boca Raton, 2013.
 
 \bibitem{Pott}
A.	Pott.: Almost perfect and planar functions, Designs, Codes and Cryptography, Vol. 78(1), pages, 41-195, 2016.
 
\bibitem{Rothaus76}
O.S. Rothaus.:
\newblock {On "bent" functions}.
\newblock  {Journal of Combinatorial Theory, Serie A Vol. 20}, pages 300--305, 1976.

 	\bibitem{W} K.S. Williams.  \newblock  Note on Cubics over $\mathbf{GF}(2^n)$ and $\mathbf{GF}(3^n)^*$.  \newblock { Journal of Number Theory.}  7: 361-365, 1975.
 	
\bibitem{Zieve13}
M.E. Zieve.:
\newblock{Permutation polynomials on $\gf_q$ induced from bijective Rédei functions on subgroups of
	the multiplicative group of $\gf_q$}, \newblock{arXiv:1310.0776}, 2013.

\bibitem{Villa} I. Villa, On APN functions $L_1(x^3)+L_2(x^9)$ with linear $L_1$ and $L_2$,  Cryptogr. Commun., Vol. 11, pp. 3–20, 2019. 


\end{thebibliography}
\end{document}